\documentclass{article}

\usepackage{amsmath,amssymb,amsthm,mismath}
\usepackage{url}
\newtheorem{theorem}{Theorem}
\newtheorem{lemma}{Lemma}

\newcommand{\ket}[1]{\mathinner{\left \lvert #1 \right\rangle}}

\title{Elementary proofs of quantum lower bounds for element distinctness, $k$-sum and $k$-distinctness}
\author{Andris Ambainis, Martins Kokainis, Jevgēnijs Vihrovs \\ \\
Centre for Quantum Computer Science, \\Faculty of Science and Technology, 
University of Latvia}

\date{}

\begin{document}

\maketitle

\begin{abstract}
We give elementary proofs of quantum query lower bounds for three well-known problems:
\begin{itemize}
    \item the $\Omega(n^{2/3})$ quantum query lower bound for element distinctness;
    \item the $\Omega(n^{k/(k+1)})$ lower bound for the $k$-sum problem, for fixed $k$ and sufficiently large groups;
    \item the $\Omega(n^{3/4-1/(4(2^k-1))})$ lower bound for the $k$-distinctness problem, even for small alphabets.
\end{itemize}
The proofs use only basic properties of quantum states and linear algebra, analyzing the differences of algorithm states on different inputs.
\end{abstract}

\section{Introduction}

Element distinctness is a well-known example of a quantum speedup. It is solvable with $O(n^{2/3})$ queries quantumly~\cite{Ambainis2007} but requires $\Omega(n)$ queries classically. The quantum algorithm played a central role in establishing quantum walks as a general paradigm for algorithm design~\cite{Szegedy2004,MNRS2011,Santha2008}.

The quantum query bound is optimal. The matching $\Omega(n^{2/3})$ lower bound was first established using the polynomial method, as a consequence of Aaronson and Shi's $\Omega(n^{1/3})$ lower bound for the collision problem~\cite{AaronsonShi2004}. Subsequent proofs used the negative-weight adversary method~\cite{Belovs2012} and the compressed oracle method~\cite{Zhandry2019,Childs2025}. These proofs involve either  substantial mathematical machinery, such as results from polynomial approximation theory, or a technically involved analysis of quantum queries.

In this note, we give an elementary proof of the $\Omega(n^{2/3})$ quantum lower bound for element distinctness. The proof is slightly more than two pages and requires only basic quantum computing and linear algebra. Our argument refines the original hybrid method of Bennett, Bernstein, Brassard, and Vazirani~\cite{BBBV1997}, which established the $\Omega(\sqrt n)$ quantum lower bound for unstructured search. The additional ingredient is a second-order difference of algorithm's states under changes to two input positions.

We find it surprising that this well-studied lower bound has such a short and transparent proof. This raises a broader question: which other quantum lower bounds have substantially simpler proofs than the currently known ones?

We then extend our argument to two other problems. The first is the $k$-sum problem, which asks whether an input contains $k$ entries whose sum equals a prescribed target. 
This is interesting as the hardest query problem with positive certificates of size $k$: any such problem is solvable with $O(n^{k/(k+1)})$ queries \cite{ChildsEisenberg2005} and, for the $k$-sum problem, Belovs and Špalek \cite{BelovsSpalek2013} showed that $\Omega(n^{k/(k+1)})$ queries are necessary, using the negative-weight adversary method. In this paper, we give an elementary proof of the $\Omega(n^{k/(k+1)})$ lower bound using higher-order differences of quantum states.

The second problem is $k$-distinctness, which asks whether some value occurs at least $k$ times in the input. Determining its quantum query complexity for fixed $k\ge3$ was a longstanding open problem. In 2012, Belovs \cite{Belovs2012a} gave a learning-graph algorithm using $O(n^{3/4-1/(4(2^k-1))})$ queries whose optimality remained open until 2026: earlier lower bounds from the polynomial method \cite{BKT2018,MTZ2020} left polynomial gaps, which Belovs \cite{Belovs2026} finally closed using a new framework. We recover the tight lower bound in an elementary way using higher-order differences, even for alphabets of only \(n\) symbols, resolving the small-alphabet question left open in \cite{Belovs2026}.

{\bf AI disclosure.} This proof was found while working with ChatGPT 6 Astra on a quantum lower bound for element distinctness in the hybrid query model of \cite{AmbainisGilyenKokainis2026}. A simplification of the proposed proof in the hybrid model led to the proof in the standard model presented in this note.

\section{Element distinctness lower bound} 
\label{sec:dist}

\begin{theorem}
\label{thm:1}
Let $\Sigma$ be an alphabet with at least $n+1$ elements.
Any bounded-error quantum algorithm for element distinctness on $n$ entries over $\Sigma$ uses $\Omega(n^{2/3})$ queries.
\end{theorem}

The proof can be extended to cover the case $|\Sigma|=n$ (which is the smallest alphabet size for which the problem is nontrivial) but becomes more complicated. 
For simplicity, we present a version which does not cover this case.

\begin{proof}
We identify $\Sigma$ with $\{0,\ldots,m-1\}$, where $m=|\Sigma|$, and use the standard value-query oracle
\begin{equation}
O_z\lvert i,b,w\rangle
=\lvert i,(b+z_i) \bmod m,w\rangle,
\label{eq:oracle}
\end{equation}
where $i\in[n]$ is the query index, $b\in\{0,\ldots,m-1\}$ is the answer register, and $w$ denotes the remaining workspace.

By deferring measurements and padding shorter executions, we can assume that the algorithm makes exactly $T$ queries, with input-independent unitaries between queries.
We write the state after $t$ queries on input $z$ as
\begin{equation}
\label{eq:0}
\ket{\psi_{t}(z)}=U_{t}O_z\ket{\psi_{t-1}(z)},
\end{equation}
where $\ket{\psi_0}$ is input-independent, $O_z$ denotes the query transformation corresponding to the input $z$ and $U_t$ are the input-independent unitaries. Let $P_i$ be the projection to the subspace spanned by basis states with query index $i$.

We choose $n+1$ distinct symbols
$a,x_1,\ldots,x_n$ uniformly at random, and define an input $x=(x_1,\ldots,x_n)$. 
For each $i$, let $x^{(i)}$ be the input obtained from $x$ by changing the value of the $i^{\rm th}$ variable to $a$. 
Crucially, $x^{(i)}$ has the same distribution as $x$, because the joint distribution of $(a,x_1,\ldots,x_n)$ is invariant under swapping $a$ and $x_i$.

We first bound the effect of replacing a uniformly random input entry on the algorithm's state.
Let
\[ \ket{\delta_{i,t}}=\ket{\psi_t(x^{(i)})}-\ket{\psi_t(x)}.
\]
\begin{lemma}
    \label{lem:1}
For $i$ chosen uniformly from $[n]$, independently of $a,x_1,\ldots,x_n$,
\[
\mathbb E_{a,x,i}\|\delta_{i,t}\|^2\le \frac{4t}{n}.
\]
\end{lemma}

\begin{proof}
Input-independent unitaries preserve distances, and the oracles $O_x$ and $O_{x^{(i)}}$ agree outside query index $i$. Thus
\[
\begin{aligned}
\|\delta_{i,t+1}\|^2
&=
\|(I-P_i)\delta_{i,t}\|^2
+
\bigl\|
O_{x^{(i)}}P_i\psi_t(x^{(i)})
-
O_xP_i\psi_t(x)
\bigr\|^2\\
&\le
\|\delta_{i,t}\|^2
+
2\|P_i\psi_t(x^{(i)})\|^2
+
2\|P_i\psi_t(x)\|^2,
\end{aligned}
\]
with the last step following from $\|a+b\|^2 \leq 2 \|a\|^2 + 2 \|b\|^2$.

For each fixed $i$, the inputs $x^{(i)}$ and $x$ have the same distribution. Moreover,
\[
\mathbb E_{a,x,i}\|P_i\psi_t(x)\|^2
=
\frac1n\,\mathbb E_{a,x}
\sum_{i=1}^n\|P_i\psi_t(x)\|^2
=\frac1n.
\]
Taking expectations therefore gives
\[
\mathbb E_{a,x,i}\|\delta_{i,t+1}\|^2
\le
\mathbb E_{a,x,i}\|\delta_{i,t}\|^2+\frac4n.
\]
Since $\delta_{i,0}=0$, the lemma follows.\end{proof}

Let 
$S_t=\mathbb E_{a,x,i}\|\delta_{i,t} \|$.
By concavity of the square-root function, Lemma \ref{lem:1} implies
\begin{equation}
S_t 
\le
\sqrt{\mathbb E_{a,x,i}\|\delta_{i,t}\|^2}
\le 2\sqrt{\frac tn}.
\label{eq:new}
\end{equation}

We now choose a uniformly random ordered pair of distinct indices $(u,v)$, independently of $a,x$, and plant a collision by replacing both $x_u$ and $x_v$ by $a$. Denote the resulting input by $y$, and define
\begin{equation}
\label{eq:2a}
\ket{\Delta_t}=\ket{\psi_t(y)}-\ket{\psi_t(x^{(u)})}
                 -\ket{\psi_t(x^{(v)})}+\ket{\psi_t(x)}.
\end{equation}
Then
\begin{equation}
\label{eq:2}
\ket{\psi_t(y)}-\ket{\psi_t(x)}=\ket{\delta_{u,t}}+\ket{\delta_{v,t}}+\ket{\Delta_t}.
\end{equation}
Let $R_t=\mathbb E_{a,x,u,v}\|\Delta_t\|$. Taking expectations, we get
\[ \mathbb E_{a,x,u,v} \| \psi_t(y)-\psi_t(x) \| \leq 2 S_t + R_t . \]
$S_t$ can be upper-bounded using (\ref{eq:new}). 
Therefore, it remains to bound $R_t$. Let $D_i=O_{x^{(i)}}-O_x$. These operators satisfy
\[
 O_y=O_x+D_u+D_v.
\qquad \|D_i\|\le2
\]
Expanding the terms in (\ref{eq:2a}) according to (\ref{eq:0}) and regrouping them gives 
\begin{equation}
\label{eq:3}
\ket{\Delta_{t+1}}
=
U_{t+1}\bigl(O_y \ket{\Delta_t}+D_v\ket{\delta_{u,t}}+D_u\ket{\delta_{v,t}}\bigr).
\end{equation}
By unitarity and the triangle inequality,
\begin{equation}
\|\Delta_{t+1}\|
\le
\|\Delta_t\|
+
2\|P_v\delta_{u,t}\|
+
2\|P_u\delta_{v,t}\|.
\label{eq:new6}
\end{equation}

For fixed $a,x,u$, the concavity of the square root gives
\[
\mathbb E_{v: v\neq u}\|P_v\delta_{u,t}\|
\le
\sqrt{\frac{\sum_{v\ne u}\|P_v\delta_{u,t}\|^2}{n-1}}
\le \frac{\|\delta_{u,t}\|}{\sqrt{n-1}}.
\]
Taking the expectation in (\ref{eq:new6}) and using the symmetry between $u$ and $v$, we get
\[
R_{t+1}
\le R_t+4\mathbb E\|P_v\delta_{u,t}\|
\le R_t+\frac{4S_t}{\sqrt{n-1}}
\le R_t+\frac{8\sqrt t}{\sqrt{n(n-1)}}.
\]
Since $R_0=0$,
\[
R_T\le
\frac{8}{\sqrt{n(n-1)}}\sum_{t=0}^{T-1}\sqrt t
\le \frac{8T^{3/2}}{\sqrt{n(n-1)}},
\]
implying that 
\begin{equation} 
\label{eq:7}
\mathbb E_{a,x,u,v} \| \psi_T(y)-\psi_T(x) \| \leq 4 \sqrt{\frac{T}{n}} +
\frac{8T^{3/2}}{\sqrt{n(n-1)}}.
\end{equation}
Every $x$ has distinct entries, whereas every $y$ has exactly one colliding pair. An algorithm with error at most $1/3$ must therefore have output probabilities differing by at least $1/3$ on each pair $(x, y)$. Since the difference between the probabilities that the algorithm outputs 1 on $x$ and $y$ is at most twice the Euclidean distance between the algorithm's final states, we must have
\[
\|\psi_T(y)-\psi_T(x)\|\ge\frac16.
\]
Thus the left-hand side of (\ref{eq:7}) is at least $1/6$, and
\[
\sqrt{\frac{T}{n}}+\frac{T^{3/2}}{n}=\Omega(1).
\]
This implies $T=\Omega(n^{2/3})$. 

\end{proof}

\section{The $k$-sum lower bound}
\label{sec:ksum}

We generalize the proof of Theorem \ref{thm:1} by replacing second-order differences with $k$th-order differences. This gives the lower bound of Belovs and
\v Spalek~\cite{BelovsSpalek2013}. 

\begin{theorem}
\label{thm:2}
Let $k\ge2$ be fixed, let $G$ be a finite abelian group with $|G|\ge n^k$,
and let $c\in G$. Any bounded-error quantum algorithm that decides whether
$x\in G^n$ contains $k$ distinct positions whose values sum to $c$ uses
$\Omega(n^{k/(k+1)})$ queries.
\end{theorem}

\begin{proof}
As in Theorem~\ref{thm:1}, consider a $T$-query algorithm. 
The queries are similar to (\ref{eq:oracle}) but use addition in $G$ instead of modular addition.
The notation $\lvert \psi_t(z) \rangle$, $O_z$, $U_t$ and $P_i$
is also the same as in the proof of Theorem \ref{thm:1}.

We may assume $T\le n$, since otherwise the claimed lower bound is immediate. We choose $x,b\in G^n$ independently and
uniformly. For $V\subseteq[n]$, let $x^V$ be obtained from $x$ by replacing
$x_i$ with $b_i$ for every $i\in V$. We define
\[
 \lvert \delta_{U,t}(x,b) \rangle
 =\sum_{V\subseteq U}(-1)^{|U|-|V|}\lvert \psi_t(x^V) \rangle.
\]
In particular, $\lvert \delta_{\varnothing,t} \rangle=\lvert \psi_t(x) \rangle$.
We first establish a
higher-order version of Lemma \ref{lem:1}.
\begin{lemma} \label{ksum:sensitivity}
For $U$ chosen uniformly among the $j$-subsets of $[n]$, independently of $x,b$,
\[
\mathbb E_{x,b,U}\|\delta_{U,t}(x,b)\|^2
\le
4^j\frac{\binom tj}{\binom nj}.
\]
Here and below, $\binom tj=0$ when $j>t$.
\end{lemma}

\begin{proof}
Let $A_{j,t}=\mathbb E_{x,b,U}\|\delta_{U,t}(x,b)\|^2$. Outside query indices
in $U$, all the oracles in the alternating sum agree. At an index $i\in U$,
the alternating sum after the query is the difference between two oracle
images of $(j-1)$st-order differences. Omitting the input-independent unitary
and using $\lVert a-b \rVert^2\le2\lVert a \rVert^2+2\lVert b \rVert^2$ therefore gives
\[
\begin{split}
 \lVert \delta_{U,t+1}(x,b) \rVert^2
 \le{}&\lVert \delta_{U,t}(x,b) \rVert^2\\
 &+2\sum_{i\in U}\left(
  \lVert P_i\delta_{U\setminus\{i\},t}(x^{\{i\}},b) \rVert^2
 +\lVert P_i\delta_{U\setminus\{i\},t}(x,b) \rVert^2\right).
\end{split}
\]
Since $x_i$ and $b_i$
are independent and uniform,
the two terms in parentheses have the same expectation. Averaging over $U$ yields
\begin{equation}
\begin{aligned}
A_{j,t+1}
&\le A_{j,t}
+\frac{4}{\binom nj}
\sum_{\substack{W\subseteq[n]\\|W|=j-1}}
\sum_{i\notin W}
\mathbb E_{x,b}\|P_i\delta_{W,t}(x,b)\|^2\\
&\le A_{j,t}
+4\frac{\binom n{j-1}}{\binom nj}A_{j-1,t}\\
&=A_{j,t}+\frac{4j}{n-j+1}A_{j-1,t}.
\end{aligned}
\label{ksum:recursion}
\end{equation}
Since $A_{0,t}=1$ and $A_{j,0}=0$ for $j\ge1$, induction and Pascal’s identity prove the lemma.
\end{proof}

Let $S_{j,t}=\mathbb E_{x,b,U}\|\delta_{U,t}(x,b)\|$, with $U$ a uniform $j$-subset. By concavity of the square root, for $1\le j\le n$ and $t\le T$,
\begin{equation}
S_{j,t}
\le
2^j\sqrt{\frac{\binom tj}{\binom nj}}
\le
2^j\left(\frac tn\right)^{j/2}.
\label{eq:new8}
\end{equation}

We now plant a solution. Choose $x\in G^n$ uniformly and an independent uniform $k$-subset $S\subseteq[n]$. Independently of $x$, choose $(a_i)_{i\in S}$ uniformly subject to $\sum_{i\in S}a_i=c$. For $V\subseteq S$, let $x^V$ now denote replacement by these $a_i$, and define $|\delta_{V,t}\rangle$ using these replacements. Put
\[
y=x^S,\qquad
|\Delta_t\rangle=|\delta_{S,t}\rangle,\qquad
R_t=\mathbb E\|\Delta_t\|.
\]
Expectations involving these planted differences are over $x,S,(a_i)_{i\in S}$.

Every proper subset of the $a_i$ consists of independent uniform elements of $G$. Also, a uniform $j$-subset of $S$ is uniform among all $j$-subsets of $[n]$. Thus, inclusion–exclusion and the triangle inequality give
\begin{equation}
\begin{gathered}
|\psi_t(y)\rangle-|\psi_t(x)\rangle
=
|\Delta_t\rangle+
\sum_{\varnothing\ne V\subsetneq S}|\delta_{V,t}\rangle,\\
\mathbb E\|\psi_t(y)-\psi_t(x)\|
\le
R_t+\sum_{j=1}^{k-1}\binom kjS_{j,t}.
\end{gathered}
\label{eq:new9}
\end{equation}
It remains to bound $R_t$.

For $i\in S$, let $D_i=O_{x^{\{i\}}}-O_x$. As before,
\[
D_i=P_iD_iP_i,\qquad
\|D_i\|\le2,\qquad
O_y=O_x+\sum_{i\in S}D_i.
\]
Expanding the definition of $\Delta_{t+1}$ gives
\begin{equation}
|\Delta_{t+1}\rangle
=
U_{t+1}\left(
O_y|\Delta_t\rangle+
\sum_{i\in S}D_i|\delta_{S\setminus\{i\},t}\rangle
\right).
\label{eq:ksum-new}
\end{equation}
By unitarity and the triangle inequality,
\begin{equation}
R_{t+1}
\le R_t+
2\mathbb E\sum_{i\in S}\|P_i\delta_{S\setminus\{i\},t}\|.
\label{eq:new10}
\end{equation}

For a uniform $i\in S$, the set $W=S\setminus\{i\}$ is a uniform $(k-1)$-subset, and $i$ is uniform outside $W$. Restricted to \(W\), the replacements $a_j$ are independent and uniform, so we can evaluate this expectation using the independent replacement vector \(b\) from Lemma \ref{ksum:sensitivity}. For fixed \(x,b,W\), the concavity of the square root gives
\[
\mathbb E_{i\notin W}\|P_i\delta_{W,t}(x,b)\|
\le
\sqrt{\frac{\sum_{i\notin W}\|P_i\delta_{W,t}(x,b)\|^2}{n-k+1}}
\le
\frac{\|\delta_{W,t}(x,b)\|}{\sqrt{n-k+1}}.
\]
Averaging and applying (\ref{eq:new10}), we obtain
\begin{equation}
R_{t+1}\le R_t+\frac{2kS_{k-1,t}}{\sqrt{n-k+1}}.
\label{eq:ksum-new1}
\end{equation}
Since $R_0=0$, (\ref{eq:new8}) implies
\begin{equation}
R_T
\le
\frac{2k}{\sqrt{n-k+1}}
\sum_{t=0}^{T-1}S_{k-1,t}
=
O_k\!\left(\frac{T^{(k+1)/2}}{n^{k/2}}\right).
\label{eq:new11}
\end{equation}
Combining (\ref{eq:new8}), (\ref{eq:new9}), and (\ref{eq:new11}), and using $T\le n$, gives
\begin{equation}
\mathbb E\|\psi_T(y)-\psi_T(x)\|
=
O_k\!\left(
\sqrt{\frac Tn}
+
\frac{T^{(k+1)/2}}{n^{k/2}}
\right).
\label{eq:new12}
\end{equation}

The input $y$ always contains a solution. Unlike in Theorem~\ref{thm:1}, $x$ need not
be a negative input, but a union bound gives
\[
 \Pr[x\text{ contains a solution}]
 \le\frac{\binom nk}{|G|}\le\frac1{k!}\le\frac12.
\]
On every pair for which $x$ is negative, the final measurement argument in Theorem \ref{thm:1} gives
\[
\|\psi_T(y)-\psi_T(x)\|\ge\frac16.
\]
Thus the left-hand side of (\ref{eq:new12}) is at least 1/12. It follows that
\[
\sqrt{\frac Tn}+\frac{T^{(k+1)/2}}{n^{k/2}}
=\Omega_k(1),
\]
and consequently $T=\Omega(n^{k/(k+1)})$.
\end{proof}

\section{The $k$-distinctness lower bound}
\label{sec:kdist}

We again use the $k^{\rm th}$-order differences of algorithm's states. The argument follows the same structure as the proof for $k$-sum, but requires a different input distribution and new estimates for the lower-order differences.

\begin{theorem}
Let $k\geq 3$ be fixed and let
$\Sigma$ be an alphabet with at least $n$ elements.
Any bounded-error quantum algorithm for $k$-distinctness on $n$ entries over $\Sigma$ uses $\Omega(n^{\frac34-\frac{1}{4(2^k-1)}})$ queries.
\end{theorem}

This gives an optimal lower bound for $k$-distinctness, even for alphabets of small size, which was left as an open question in \cite{Belovs2026}.

\begin{proof}
We may assume that $n$ is sufficiently large compared to $k$.
Set
\begin{equation*}
m=\left\lfloor\frac{n}{k^2}\right\rfloor,
\qquad s=n-m\sum_{r=2}^{k-1}r.
\end{equation*}
We choose $x$ uniformly among inputs having exactly $m$ classes of equal elements of each size $2,\ldots,k-1$ and $s$ singletons (elements $x_i$ that are distinct from any other $x_j$). Let $a$ be a uniformly random symbol that does not appear in $x$. Let $L(x)$ be the set of singleton positions. For $V\subseteq L(x)$, let $x^V$ be obtained by replacing every entry indexed by $V$ by $a$.

We have
\begin{equation*}
s\ge n/2,\qquad m=\Theta(n).
\end{equation*}


As in the proof of Theorem \ref{thm:2}, we define
\begin{equation*}
\ket{\delta_{V,t}}=\sum_{W\subseteq V}(-1)^{|V|-|W|}\ket{\psi_t(x^W)}.
\end{equation*}
For $1\le j\le k$, define expected differences
\begin{equation*}
A_{j,t}=\mathbb E_{\substack{x,a,S\\|S| = j}}
\norm{\delta_{S,t}}^2.
\end{equation*}
In this and the following expectations, $x$ and $a$ are sampled as above, and $S$ is a uniform subset of $L(x)$ of the indicated size. 

\begin{lemma} \label{lem:diff} For all $j = 1, \ldots, k-1$, we have
\begin{equation*}
    A_{j,t}=O\!\left(\left(\frac t n\right)^{2-2^{1-j}}\right).
\end{equation*}
\end{lemma}

\begin{proof}
Let $j \geq 2$ and let $S\subseteq L(x)$ be a set with $|S| = j$.
For each $i\in S$, we define
\begin{align*}
\ket{G_i}&=P_i\sum_{V\subseteq S,\, i \in V}
(-1)^{j-|V|}\ket{\psi_t(x^{V})}\\
\ket{H_i}&=P_i\ket{\delta_{S\setminus\{i\},t}},
\end{align*}
with $P_i$ being the projection to the basis states where the query index register contains $i$.
Then $P_i\ket{\delta_{S,t}}=\ket{G_i}-\ket{H_i}$. 
In every input appearing in $\ket{G_i}$,
the value at position $i$ is $a$, whereas in every input appearing
in $\ket{H_i}$, the value there is $x_i$. Applying a query to $\ket{G_i}-\ket{H_i}$ changes it 
to $Q_a\ket{G_i}-Q_{x_i}\ket{H_i}$, where $Q_b$ is the answer-register
shift by $b$.

For any unitaries $Q,Q'$, the triangle inequality gives
\begin{equation*}
\norm{Q\ket{G}-Q'\ket{H}}^2-\norm{\ket{G}-\ket{H}}^2  \le(\norm{G}+\norm{H})^2-(\norm{G}-\norm{H})^2 =4\norm{G}\norm{H}.
\end{equation*}
All components outside $S$ preserve their norms, as does the common
input-independent unitary. Consequently,
\begin{equation} \label{eq:gain}
A_{j,t+1}-A_{j,t}
\le4\mathbb E_{\substack{x,a,S\\|S|=j}}\sum_{i\in S}\norm{G_i}\norm{H_i} \leq 4\sqrt{\smash[b]{\mathbb E_{\substack{x,a,S\\|S|=j}}\sum_{i\in S}\norm{G_i}^2 \cdot \mathbb E_{\substack{x,a,S\\|S|=j}}\sum_{i\in S}\norm{H_i}^2}},
\end{equation}
where the last inequality is by Cauchy-Schwarz.
We bound the two expectations separately.

\begin{enumerate}
    \item For the $H_i$ terms, take expectation over $W=S\setminus\{i\}$. The remaining
index $i$ is uniform in $L(x)\setminus W$, so
\begin{equation}
\begin{aligned}
\mathbb E_{x,a,S}\sum_{i\in S}\|H_i\|^2
&=\frac{j}{s-j+1}\,
  \mathbb E_{x,a,W}
  \sum_{i\in L(x)\setminus W}\|P_i\ket{\delta_{W,t}}\|^2\\
&\le \frac{j}{s-j+1}A_{j-1,t}.
\end{aligned}
\label{eq:B}
\end{equation}

\item For the $G_i$ terms, choose $i$ uniformly from $S$, and then choose $V$ uniformly among the $2^{j-1}$ subsets of $S$ containing $i$. By Cauchy–Schwarz,
\[
\begin{aligned}
\mathbb E_{x,a,S}\sum_{i\in S}\|G_i\|^2
&\le
2^{j-1}\mathbb E_{x,a,S}
\sum_{i\in S}\sum_{\substack{V\subseteq S\\i\in V}}
\|P_i\ket{\psi_t(x^V)}\|^2\\
&=
j\,4^{j-1}\mathbb E_{x,a,S,i,V}\|P_i\ket{\psi_t(x^V)}\|^2.
\end{aligned}
\]
We bound the last expectation by conditioning on $|V|=r$. Write $z=x^V$.

If $r=1$, then $z$ has $s$ singleton positions and, conditional on $z$, the index $i$ is uniform among them. Therefore,
\[
\mathbb E\!\left[\|P_i\ket{\psi_t(z)}\|^2\mid |V|=1\right]
=
\mathbb E\!\left[
\frac1s\sum_{i\in L(z)}\|P_i\ket{\psi_t(z)}\|^2
\,\middle|\, |V|=1
\right]
\le \frac1s.
\]

If $r\ge2$, then $z$ has $m+1$ equality classes of size $r$. By symmetry, conditional on $z$ and $|V|=r$, the newly created class $V$ is uniform among these classes, and $i$ is uniform within it. Thus $i$ is uniform among the $r(m+1)$ positions belonging to classes of size $r$, giving
\[
\mathbb E\!\left[\|P_i\ket{\psi_t(z)}\|^2\mid |V|=r\right]
\le \frac1{r(m+1)}.
\]

Since $s,m=\Theta(n)$, each conditional expectation is $O(1/n)$. Averaging over $r$, and recalling that $j\le k-1$ with $k$ fixed, we obtain
\[
\mathbb E_{x,a,S}\sum_{i\in S}\|G_i\|^2=O(1/n).
\]\end{enumerate}

Consequently, we can now upper bound \eqref{eq:gain}:
\begin{equation} \label{eq:rec}
    A_{j,t+1} - A_{j,t} = O\!\left(\sqrt{\frac{1}{n} \cdot \frac{j}{s-j+1} A_{j-1,t}} \right) = O\!\left(\frac{1}{n} \sqrt{A_{j-1,t}} \right).
\end{equation}
It remains to solve this recurrence, for which we use induction.
The case $j = 1$ follows by the same argument as Lemma \ref{lem:1}, with $i$ uniform over $L(x)$.
Define $\alpha_1=1$, $\alpha_j=1+\frac{\alpha_{j-1}}2$.
Then $\alpha_j=2-2^{1-j}$.
Inductively, if the estimate holds for $j-1$, summing the upper bound~\eqref{eq:rec} gives
\begin{equation*}
A_{j,T}
= O\!\left(\frac1n\sum_{t=0}^{T-1}
(t/n)^{\alpha_{j-1}/2}\right)
=O\!\left((T/n)^{1+\alpha_{j-1}/2}\right)
= O\!\left((T/n)^{2-2^{1-j}}\right),
\end{equation*}
which proves the lemma.
\end{proof}

For $1\le j\le k$, define the expected norms
\begin{equation*}
S_{j,t}=\mathbb E_{\substack{x,a,S\\|S|=j}}\norm{\delta_{S,t}}.
\end{equation*}
By concavity of the square-root function and Lemma~\ref{lem:diff},
for $1\le j\le k-1$,
\begin{equation}\label{eq:first-moments}
S_{j,t}\le\sqrt{A_{j,t}}
=O\!\left((t/n)^{1-2^{-j}}\right).
\end{equation}
Now, we upper bound the $k$th-order expected norm.

\begin{lemma} \label{lem:final-diff}
    We have that
    \begin{equation*}
        S_{k,t} = O\!\left(\frac{t^{2-2^{1-k}}}{n^{3/2-2^{1-k}}}\right).
    \end{equation*}
\end{lemma}

\begin{proof}
Let $S\subseteq L(x)$ have size $k$, put $y=x^S$, and define $D_i=O_{x^{\{i\}}}-O_x$ for $i\in S$. The same expansion as in (\ref{eq:ksum-new}) gives
\[
\ket{\delta_{S,t+1}}
=
U_{t+1}\left(
O_y\ket{\delta_{S,t}}
+\sum_{i\in S}D_i\ket{\delta_{S\setminus\{i\},t}}
\right).
\]
Since $D_i=P_iD_iP_i$ and $\|D_i\|\le 2$, unitarity and the triangle inequality imply
\[
S_{k,t+1}
\le S_{k,t}
+2\mathbb E_{x,a,S}\sum_{i\in S}
\|P_i\ket{\delta_{S\setminus\{i\},t}}\|.
\]

We now apply the averaging argument used to obtain (\ref{eq:ksum-new1}), with $L(x)$ in place of $[n]$. If $S$ is uniform among the $k$-subsets of $L(x)$ and $i$ is uniform in $S$, then $W=S\setminus\{i\}$ is a uniform $(k-1)$-subset of $L(x)$, and, conditional on $x,a,W$, the index $i$ is uniform in $L(x)\setminus W$. Consequently,
\[
\mathbb E_{i\in L(x)\setminus W}\|P_i\ket{\delta_{W,t}}\|
\le \frac{\|\delta_{W,t}\|}{\sqrt{s-k+1}},
\]
and hence
\[
S_{k,t+1}
\le S_{k,t}
+\frac{2k}{\sqrt{s-k+1}}S_{k-1,t}.
\]
Let $\alpha=\alpha_{k-1}=2-2^{2-k}$.
Since $S_{k, 0}=0$, summing this recurrence and
using~\eqref{eq:first-moments} gives
\begin{align*}
S_{k,T}
&\le\frac{2k}{\sqrt{s-k+1}}\sum_{t=0}^{T-1}S_{k-1,t}
=O\!\left(\frac1{\sqrt n}\sum_{t=0}^{T-1}(t/n)^{\alpha/2}\right)\\
&=O\!\left(\frac{T^{1+\alpha/2}}{n^{(1+\alpha)/2}}\right)
=O\!\left(\frac{T^{2-2^{1-k}}}{n^{3/2-2^{1-k}}}\right),
\end{align*}
which proves the lemma.
\end{proof}

We now complete the proof of the theorem.
Choose $S$ uniformly among the $k$-subsets of $L(x)$, and put $y=x^S$. The same inclusion–exclusion and triangle-inequality argument as in (\ref{eq:new9}) gives
\begin{equation}
\mathbb E_{x,a,S}\|\psi_T(y)-\psi_T(x)\|
\le
\sum_{j=1}^{k}\binom{k}{j}S_{j,T}.
\label{eq:distance}
\end{equation}
Here we use that a uniformly chosen $j$-subset of $S$ is uniform among the $j$-subsets of $L(x)$.

We may assume $T\le n$, since otherwise the claimed lower bound is immediate.
By~\eqref{eq:first-moments}, $S_{j,T}=O(\sqrt{T/n})$ for
$j=1,\ldots,k-1$.
Let $\alpha=\alpha_{k-1}=2-2^{2-k}$.
Combining~\eqref{eq:distance} with Lemma~\ref{lem:final-diff}, we get
\begin{equation*}
\mathbb E_{\substack{x,a,S\\|S|=k}}\norm{\psi_T(y)-\psi_T(x)}
=O\!\left(\sqrt{\frac Tn}+
\frac{T^{1+\alpha/2}}{n^{(1+\alpha)/2}}\right).
\end{equation*}

Every sampled $x$ is a negative input and every sampled $y$ is a
positive input. As in Theorem~\ref{thm:1}, an algorithm with error probability at most
$1/3$ must satisfy $\norm{\psi_T(y)-\psi_T(x)}\ge1/6$ on every such pair.
Consequently,
\begin{equation*}
\sqrt{\frac Tn}+\frac{T^{1+\alpha/2}}{n^{(1+\alpha)/2}}=\Omega(1),
\end{equation*}
which implies
\begin{equation*}
T=\Omega\!\left(n^{(\alpha+1)/(\alpha+2)}\right),
\end{equation*}
where
\begin{equation*}
\frac{\alpha+1}{\alpha+2}
=\frac{3-2^{2-k}}{4-2^{2-k}}
=\frac34-\frac1{4(2^k-1)}.\qedhere
\end{equation*}
\end{proof}

\section{Conclusion}

The adversary~\cite{Ambainis2002,HoyerLeeSpalek2007} and polynomials~\cite{BealsEtAl2001} methods have been very successful in proving tight quantum lower bounds for a wide range of problems. Our results suggest that elementary arguments based on comparisons of quantum states can also lead to sharp bounds beyond unstructured search. It is natural to ask how far such arguments can be developed and which other lower bounds admit similarly simple proofs.

For collision-like problems, natural extensions include inputs with prescribed multiplicity profiles and the task of finding several collisions. Our \(k\)-distinctness proof uses inputs containing many equality classes of every smaller size. It would be interesting to understand how the argument changes when some of these classes are scarce or absent, and whether it can give bounds on the success probability of algorithms required to find multiple collisions.

More broadly, one may ask whether higher-order differences yield a general approach to problems with planted certificates, including copies of fixed subgraphs. 
One well-known target is the triangle problem. The best upper bound for the triangle problem is $O(n^{5/4})$~\cite{LeGall2014,CaretteLauriereMagniez2017} where $n$ is the number of vertices and the best lower bound is the very recent bound of $\Omega(n^{13/12})$ \cite{CornelissenGilaniZhou2026}. Recent lower bounds for cliques and complete bipartite graphs provide further concrete targets \cite{GilaniZhou2026}. A central question is how the occurrence and overlap of partial certificates affect the growth of these differences.

{\bf Acknowledgments.} A.A. and M.K. are supported by the HQCC project funded under QuantERA II ERANET Cofund scheme of Horizon 2020 and by the University of Latvia project ``Latvian quantum technology excellence and competence centre''. J.V.~is supported by the 1.1.1.9 Research application No 1.1.1.9/LZP/2/25/206 of the Activity ``Post-doctoral Research'' ``Practical applications and limitations of quantum search''.

\bibliographystyle{plain}
\bibliography{bib}

\end{document}